\documentclass[letterpaper, 10 pt, conference]{ieeeconf}  % Comment this line out if you need a4paper
\usepackage{amsmath}
\usepackage{amsthm}
\usepackage{amssymb}
\usepackage[utf8]{inputenc}
\usepackage[T1]{fontenc}
\usepackage{graphicx}
\usepackage{enumerate}
\usepackage[backend=biber,style=ieee,citestyle=ieee]{biblatex}
\usepackage{xcolor}
\usepackage{comment}
\usepackage[caption=false]{subfig}

\AtBeginBibliography{\footnotesize}  %\small %\footnotesize %\scriptsize
\newtheorem{theorem}{Theorem}
\newtheorem{proposition}[theorem]{Proposition}

\theoremstyle{definition}
\newtheorem{definition}{Definition}%[section]

\IEEEoverridecommandlockouts                              % This command is only needed if you want to use the \thanks command

\title{\LARGE \bf Safety Embedded Control of Nonlinear Systems via Barrier States*}

\author{Hassan Almubarak$^{1,4}$, Nader Sadegh$^{2}$ and Evangelos A. Theodorou$^{3}$ 
\thanks{*This research was supported by the NSF CPS grant no. 1932288.} 
\thanks{$^{1}$School of Electrical and Computer Engineering} 
\thanks{$^{2}$The George W. Woodruff School of Mechanical Engineering}
\thanks{$^{3}$The Daniel Guggenheim School of Aerospace Engineering}
\thanks{Georgia Institute of Technology, Atlanta, GA 30332, USA}
\thanks{$^{4}$ Department of Control and Instrumentation Engineering, King Fahd University of Petroleum \& Minerals, Dhahran 31261, Saudi Arabia}
\thanks{{\tt\footnotesize halmubarak, sadegh, evangelos.theodorou@gatech.edu}}

     }   
\begin{document}

\maketitle
\thispagestyle{empty}
\pagestyle{empty}

%%%%%%%%%%%%%%%%%%%%%%%%%%%%%%%%%%%%%%%%%%%%%%%%%%%%%%%%%%%%%%%%%%%%%%%%%%%%%%%%
\begin{abstract}
 %Many of today's rapidly growing engineering technologies are accompanied with highly challenging problems, and safety is, undeniably, a crucial one. 
 In many safety-critical control systems, possibly opposing safety restrictions and control performance objectives arise. To confront such a conflict, this letter proposes a novel methodology that embeds safety into stability of control systems. The development enforces safety by means of barrier functions used in optimization through the construction of \textit{barrier states} (BaS) which are \textit{embedded} in the control system's model. As a result, as long as the equilibrium point of interest of the closed loop system is asymptotically stable, the generated trajectories are guaranteed to be safe. Consequently, a conflict between control objectives and safety constraints is substantially avoided. To show the efficacy of the proposed technique, we employ barrier states with the simple pole placement method to design safe linear controls. Nonlinear optimal control is subsequently employed to fulfill safety, stability and performance objectives by solving the associated Hamilton-Jacobi-Bellman (HJB) which minimizes a cost functional that can involve the BaS. Following this further, we exploit optimal control with barrier states on an unstable, constrained second dimensional pendulum on a cart model that is desired to avoid low velocities regions where the system may exhibit some controllability loss and on two mobile robots to safely arrive to opposite targets with an obstacle on the way.
\end{abstract}
%%%%%%%%%%%%%%%%%%%%%%%%%%%%%%%%%%%%%%%%%%%%%%%%%%%%%%%%%%%%%%%%%%%%%%%%%%%%%%%%
\section{Introduction}
Control theory has been a central element in today's fast growing, interdisciplinary technologies from simple decision making problems to terrifically complex autonomous systems. Undoubtedly, advancements in technologies are faced with unprecedented challenges. One vital challenge is safety. % which calls for safely controlling systems to avoid intolerable losses. 
Nonetheless, conflicting safety restrictions and control objectives potentially appear in safety-critical control systems. The main goal of this letter is to develop a control design methodology that satisfies safety constraints and evades the possible conflict between safety constraints and control objectives. The objective is achieved by embedding the safety constraints in the system's model through barrier states (BaS) which are provided to the control law used to attain performance objectives. %This is a multi-objective control problem, where safety and optimal stability performance objectives are desired.  
%Although, safety-critical control systems have been a main focus in control and gained more attention recently, the need for efficiently and safely controlling a dynamical system has not been well addressed.
%Incorporating safety into the optimal control design allows performing safe control tasks efficiently and effectively
%Due to the nature of barrier functions and safety conditions, the resulting control system is nonlinear, even for linear systems. 

For a dynamical system, safety is classically verified by invariance of the set of permitted states. The set of permitted states is (forward) invariant if at some point in time it contains the system's state then it contains the state for all (future) times \cite{blanchini1999set}. To formally prove invariance and verify safety, barrier certificates were introduced in the control literature  \cite{prajna2004safety} (for a brief historical overview, see \cite{ames2019CBF-Theo&App}).
Extending the idea to control dynamical systems, a set is said to be controlled invariant, also called viable, if for any initial condition in the set, the associated trajectory is forced to be in the set for all future times using a proper control. Influenced by barrier certificates and control Lyapunov functions (CLFs), control barrier functions (CBFs) were introduced in \cite{wieland2007constructive}, which were further developed in \cite{romdlony2014uniting,ames2014control,romdlony2016stabilization,ames2016CBF-forSaferyCritControl}. CBFs can be looked at as state-dependent hard input constraints which are commonly used in quadratic programming (QP) fashion to produce a safe controller \cite{ames2014control,ames2016CBF-forSaferyCritControl,ames2016CBF-forSaferyCritControl}. CBFs are becoming increasingly popular in multi-objective control due to their ability of rendering sets invariant and flexible unification with CLFs. %Control barrier functions have been commonly unified with control Lyapunov functions (CLFs) to achieve multi-objective control.
%Although CLFs have been successfully used to design stabilizing controllers, there is no systematic methodology to find a CLF for general nonlinear systems \cite{primbs1999nonlinear}. 
Nonetheless, to avoid conflicts in the CLF-CBF QP, one of the conditions needs to be relaxed. For a complete review on CBFs, the reader may refer to \cite{ames2019CBF-Theo&App}.

In this letter, inspired by adopting barrier functions to establish inequality constraints forming CBFs, we utilize barrier functions to construct \textit{barrier states}. Those barrier states (BaS) create barriers in the state space forcing the search of a stabilizing feedback control law to the set of safe controls. Specifically, the barrier states are appended to the model of the safety-critical dynamical system to generate a system that is safe if the equilibrium point of interest is asymptotically stable. %The barrier states incorporate some design parameters to fine-tune the rate of changes of those barrier states which influence the design of the feedback controller.
Therefore, since safety and stability have been unified, designing a stabilizing controller for the new model means a stabilizing safe control for the original dynamical system. In other words, the safety property is embedded in the stability of the system as the feedback controller is a function of the system's states and the barrier states, hence the name \textit{safety embedded control}. The proposed method is general enough to be used with any valid barrier function. For standard barrier functions, such as Log or inverse, we show that it is possible to generate smooth or even analytic state equations for barrier states if desired. Unlike control barrier functions, no explicit knowledge of the relative degree of the system with respect to the output describing the safe region is required.

The letter is organized as follows. Section \ref{Sec: problem statement} presents the problem statement. In Section \ref{Sec: Barrier States for Safe Stabilization}, utilizing barrier functions, we develop barrier states used to embed safety in the stabilization problem. Subsequently, we implement the developed concept to design simple safe linear controls in Section \ref{Safety Embedded Linear Control}. In Section \ref{Sec: EMBEDDED SAFETY OPTIMAL CONTROL}, we employ barrier states in the context of constrained optimal control to meet safety and performance objectives through minimizing an infinite horizon cost functional. In Section \ref{Sec: NUMERICAL IMPLEMENTATION}, we show the effectiveness of the proposed technique through safely stabilizing an unstable pendulum on a cart model where it is desired to avoid low velocity regions near the $90^{\text{o}}$ angle. Moreover, a multi-robot example is used where two simple mobile robots are to go to specific locations safely, i.e. without colliding while avoiding obstacles on the way. Lastly, conclusion remarks are driven in Section \ref{Sec: CONCLUSIONS}.

\section{Problem Statement} \label{Sec: problem statement}
Consider the nonlinear control-affine dynamical system
\begin{equation} \label{Dynamical System}
    \dot{x}= f(x)+ g(x) u
\end{equation}
where $x \in \mathcal{D} \subset \mathbb{R}^n$,\ $u \in \mathcal{U} \subset \mathbb{R}^m$, $f: \mathbb{R}^n \rightarrow \mathbb{R}^n$ and $g: \mathbb{R}^n  \rightarrow \mathbb{R}^{n\times m}$ are continuously differentiable and $f(0)=0$, without loss of generality. We wish to formulate a continuous stabilizing feedback controller $u=K(x)$ that renders the 
%such that the resulting closed-loop system meets the \emph{safety} requirements to be specified. The safe region is a 
nonempty open subset $\mathcal{C}$ of $\mathcal{D}$ defined as $\mathcal{C}=\{x\in \mathcal{D}: h(x)> 0\}$ controlled invariant, where $h : \mathcal{D} \rightarrow \mathbb{R}$ is a continuously differentiable scalar valued function that represents the safe set. The set $\mathcal{D}$ represents the domain of operation which will be further prescribed later in the letter. The set $\mathcal{C}$, referred to as the safe set, is said to be forward invariant with respect to the closed-loop system $\dot{x}=f(x)+g(x)K(x)$ if $x(0)\in \mathcal{C}$ implies $x(t) \in \mathcal{C}$, $\forall t \geq 0$. 
\begin{definition} \label{Def: safety and safe condition}
The continuous feedback controller $u=K(x)$ is said to be safe if $\mathcal{C}$ is forward invariant with respect to the resulting closed-loop system $\dot{x}=f(x)+g(x)K(x)$. That is, 
\begin{equation} \label{safety constraint}
    h(x(t))>0 \ \forall t \geq 0 ;\ x(0) \in \mathcal{C}
\end{equation}
We refer to this as the safety condition.
\end{definition}
Enforcement of the safety constraint may be facilitated by means of a smooth scalar valued function $B: \mathcal{C}\rightarrow \mathbb{R}$ known as a {\em barrier function} (BF) in the optimization literature. Popular barrier functions include the inverse barrier function (a.k.a. Carroll barrier) $B(\eta)=1/\eta$ and the logarithmic barrier function $B(\eta)=\log (\frac{1+\eta}{\eta})$. The main properties of those types of barrier functions are that they blow up at the boundaries of the complement of $\mathcal{C}$, i.e. $\lim_{\eta \rightarrow 0} B(\eta)=\infty$, $\lim_{\eta \rightarrow \infty} B(\eta)=0$ and $\inf_{\eta \in \mathbb{R}^+} B(\eta) \geq 0$. Another choice that has the advantage of being analytic with a known power series expansion is $B(\eta)=\tanh^{-1} (e^{-\eta})$. The composite BF corresponding to $h$ that defines $\mathcal{C}$ is defined to be
$\beta(x):=B \circ h (x)$. Note that $\beta(x)\rightarrow \infty$ if and only if $h(x)\rightarrow 0$. The following Proposition follows from Definition \ref{Def: safety and safe condition} and the choice of BFs:
\begin{proposition}
For the control system in \eqref{Dynamical System}, the feedback controller $u=K(x)$ satisfies the safety condition \eqref{safety constraint} if and only if $\beta (x(0)) < \infty$ implies $\beta(x(t))<\infty \ \forall t >0$.
\label{safety prop 1}
\end{proposition}
%The discussion and use of barrier functions in dynamical systems to show forward invariance is well established in the literature and is omitted \cite{prajna2003barrier,prajna2004safety,wieland2007constructive,ames2016CBF-forSaferyCritControl,ames2019CBF-Theo&App}. 

% \section{Safety Embedded Stabilization} 
\section{Barrier States for Safety Embedded Stabilization} \label{Sec: Barrier States for Safe Stabilization}
%One of the difficulties associated with integrating safety constraints into the feedback control design is the discontinuity of $\beta(x)$ at the safe region's boundary $\partial \mathcal{C}$, which violates the usual continuity assumptions many control design methods require. As will be shown next, it is possible to embed the safety constraint without introducing any discontinuities. 
As possibly conflicting safety constraints and control performance objectives need to be avoided in safety-critical control, current multi-objective frameworks, e.g. CBF-CLF safe stabilization framework, avoid possible conflicts by relaxing one of the constraints to ensure feasibility of the solution. This may result in an undesirable performance. To untangle such a problem, we propose a provable safe control technique that satisfies the safety constraints and the performance objectives simultaneously with no relaxation. 

For a barrier function $\beta(x)=B(h(x))$ where $h(x)$ defines the safe set $\mathcal{C}$, the idea is to augment the open loop system with a new state variable $z$ that is related to the recentered barrier function \cite{wills2004barrier}, $\beta(x) - \beta(0)$ to ensure that $z=0$ is an equilibrium state. %The shift ensures that the augmented system has the origin as its equilibrium point. 
If we simply set $z$ to $\beta(x)-\beta(0)$, the resulting state equation would be
\begin{align*}
    \dot{\beta}(x) &=B'(h(x))(L_fh(x)+L_gh(x)u) \\
    & =  \phi_0(\beta(x)) (L_fh(x)+L_gh(x)u)
\end{align*}
where $\phi_0=B'\circ B^{-1}$ and $B^{-1}$ is the inverse of the BF. To safely stabilize the origin of \eqref{Dynamical System}, however, we need to ensure stabilizability of the origin of the augmented system. The main issue with augmenting $\dot{\beta}$ as the state equation is that $\beta(x)$ becomes a redundant state and the resulting augmented system may not be %controllable or even 
stabilizable in some cases. Fortunately, this issue can be resolved by perturbing the barrier state equation through an auxiliary function $\phi_1$ that ensures stabilizability of the origin of the augmented system without affecting the safety guarantees. More specifically, we modify the state equation for $z$ according to  
\begin{equation} \label{zdot-barrier state}
    \dot{z} = \phi_0 (z+\beta_0) \dot{h}(x) -\gamma \phi_1(z+\beta_0,h(x))
\end{equation}
where $\beta_0 = \beta(0)$, $\dot{h}(x)=L_fh(x)+L_gh(x)u$, $\gamma \in \mathbb{R}^+$ and $\phi_1(\zeta,\eta)$ is an analytic function of two variables satisfying 
$$ \phi_1(\beta(x),h(x))=0 \ \text{and} \ \frac{\partial \phi_1}{\partial \zeta}(\beta_0,h(0)) > 0 $$
It can be seen that for all three BFs mentioned earlier, the function $\phi_0:=B'\circ B^{-1}$ is analytic with no singularities. Moreover, it is worth noting that both $\phi_0$ and $\phi_1$ are formulated independently of the system's model based on the barrier function and $h(x)$. For instance, if $\zeta = B(\eta)=1/\eta$, then $\phi_0(\zeta)=B'(B^{-1}(\zeta))=-\zeta^2$. Furthermore, letting $\phi_1(\zeta,\eta)=\zeta (\eta \zeta-1)$ satisfies the first condition $\phi_1(\zeta,\eta)=0$ along $\zeta = \beta(x)$ and $\eta=h(x)$ since $\beta(x)h(x)=1$ and the second condition
$
\frac{\partial \phi_1}{\partial \zeta}=2 \zeta \eta-1|_{\beta_0,h(0)}=2 \beta_0h(0)-1 =1>0
$.
Table \ref{tab:phi} provides the explicit expressions for $\phi_0$ and possible $\phi_1$ for most commonly used barrier functions. It should be noted that the positive scalar $\gamma$ is a design parameter that adjusts the \textit{rate} at which $z$ returns to $\beta(x)-\beta_0$ if it deviates from it at any instant of time. As a consequence, $\gamma$ has some influence on the design of the safe feedback control gains as we will demonstrate in the simulation examples.  
\begin{table}[]
    \centering
    \begin{tabular}{||c|c|c||}
    \hline\hline
        $\zeta$=B($\eta$) & $\phi_0(\zeta)$ & $\phi_1(\zeta,\eta)$ \\
        \hline
        $\log (\frac{1+\eta}{\eta})$  & $-4\sinh^2(\zeta/2)$ & $\eta(e^\zeta-1)^2-e^\zeta+1$\\
        \hline
        $1/\eta$ & $-\zeta^2$ & $\eta \zeta^2-\zeta$\\
        \hline
        $2\tanh^{-1} (e^{-\eta})$ & $-\sinh(\zeta)$ & $\tanh(\zeta/2)-e^{-\eta}$ \\
        \hline\hline
    \end{tabular}
    \caption{Functions $\phi_0$ and $\phi_1$ for three barrier functions.}
    \label{tab:phi}
        \vspace{-1cm}
\end{table}
The following proposition shows that if $z(0)$ is defined properly, then $z(t)=\beta(x)-\beta_0$ and $\phi_1(z+\beta_0,h(x))=0$, $\forall t\geq 0$, implying that boundedness of $z$ guarantees satisfaction of the safety constraint.
\begin{proposition}
Suppose that $z(0)=\beta(x(0))-\beta(0)$ and $\beta(x(0))<\infty$. Then, the auxiliary state variable $z(t)$ generated by the perturbed state equation in \eqref{zdot-barrier state} along the trajectories of \eqref{Dynamical System} is bounded if and only if $\beta(x(t))$ is bounded $\forall t$.
\label{safety prop 2}
\end{proposition}
\begin{proof}
We prove the Proposition by establishing that 
\begin{equation} z(t)=\beta(x(t))-\beta(0), \ \forall t \geq 0\end{equation}
Subtracting $\dot{\beta}(x) =\phi_0(\beta(x))\dot{h}(x)$ from both sides of \eqref{zdot-barrier state}, we have
$
\dot{\tilde{z}} (t) = \tilde{\phi}_0 (\tilde{z},x)\dot{h}(x) -\gamma \phi_1(\tilde{z}+\beta(x),h(x))
$
where $\tilde{z}=z+\beta(0)-\beta(x)$ and 
$
\tilde{\phi}_0 (\tilde{z},x)=\phi_0(\tilde{z}+\beta(x))-\phi_0(\beta(x) )
$.
The preceding differential equation has an equilibrium point at $\tilde{z}=0$ since $\phi_1(\beta(x),h(x))=0$ and $\tilde{\phi}_0 (0,x)=0$. Thus $\tilde{z}(t)=0$, or equivalently $z(t)=\beta(x(t))-\beta(0)$, $\forall t \geq 0$, provided that $\tilde{z}(0)=z(0)+\beta(0)-\beta(x(0))=0$ as assumed by the hypothesis.
\end{proof}
If desired, multiple constraints can be combined to form one BF, as done in the optimization literature, to create a single BaS. Creating a single BaS is favorable in many applications and helps avoiding adding many nonlinear state equations to the safety embedded model which may increase the complexity of the controller. A single BaS that represents different constraints, however, may be highly nonlinear and may be more difficult to use to design a safety embedded control. Additionally, some flexibility on choosing design parameters for each constraint, such as penalization of the barrier states in the case of optimal control, will be lost since we will have only one barrier state. It is important to mention that in some applications, one could represent multiple constraints with one function $h(x)$ defining the safety set and then use a single BaS.

For $q$ constraints, let $\beta(x)=\sum_{i=1}^{q} B(h_i(x))$. Then,
$$
\dot{\beta}(x)=\sum_{i=1}^{q} B'\circ B^{-1} \Big(z+\beta_0 - \sum_{j=1, j\neq i}^q B\big( h_j(x) \big) \Big) \dot{h}_i
$$
%$$
%\dot{\beta}(x)=\sum_{i=1}^{q} B'\circ B^{-1} \Big(z+\beta_0 - \beta + B\big( h_i(x) ) \big) \Big) \dot{h}_i
%$$
and therefore the BaS can be found to be
\begin{align}\begin{split} \label{multi-BaS}
    \dot{z}=\sum_{i=1}^{q} & \Bigg[   \phi_0 \Big(  z+\beta_0 - \sum_{j=1, j\neq i}^q B\big( h_j(x) \big) \Big) \dot{h}_i  \\
    &  -\gamma_i \phi_1 \Big( z+\beta_0 - \sum_{j=1, j\neq i}^q B\big( h_j(x) \big) , h_i \Big) \Bigg]
\end{split} \end{align}
In the simulation examples, we use a single BaS to represent multiple constraints for one problem and we use multiple ones in the other to validate the proposed technique.

Now, we are in a position to create the safety embedded model.  By defining $z=[z_1, \dots, z_q]^{\rm{T}}$, if more than one BaS is used, and augmenting it to the system \eqref{Dynamical System}, we get
\begin{equation} \begin{split}\label{eq:Augmented_Representation} 
    & \dot{x}= f(x)+ g(x) u \\ 
    & \dot{z}= f_b(x,z)+ g_b(x,z) u
\end{split}\end{equation}
where $f_{b}(x,z)$ and $g_{b}(x,z)$ are defined according to \eqref{zdot-barrier state} or \eqref{multi-BaS}.
%\begin{align*}  \label{fb anf gb}
%    f_{b}(x,z) &= -\gamma \phi_1(z+\beta_0,h(x))+\phi_0 (z+\beta_0)L_fh(x) \\
%    g_{b}(x,z) &= \phi_0 (z+\beta_0) L_gh(x)
%\end{align*}
By the definition of $\phi_0$ and the restrictions imposed on $\phi_1$, it follows that both $f_b$ and $g_b$ are analytic if $f$, $g$, and $h$ are analytic, and the subsystem $\dot{z}= f_b(x,z)+ g_b(x,z) u$ is stabilizable at the origin. Hence, the combined system, which can be more compactly described as
\begin{equation} \begin{split} \label{new system, safety augmented}
    & \dot{\bar{x}}= \bar{f}(\bar{x})+ \bar{g}(\bar{x}) u \\ 
\end{split} \end{equation}
where $\bar{x}=\begin{bmatrix} x \\ z \end{bmatrix}, \bar{f}=\begin{bmatrix} f \\ f_b\end{bmatrix}$ with $\bar{f}(0)=0$ and $\bar{g}=\begin{bmatrix} g \\ g_b\end{bmatrix}$, preserves the continuous differentiability and stabilizability of the original control system \eqref{Dynamical System}. Therefore, the safety constraint is \textit{embedded} in the closed-loop system's dynamics and stabilizing the safety embedded system \eqref{new system, safety augmented} implies enforcing safety for the safety-critical system \eqref{Dynamical System}, i.e. forward invariance of the safe set $\mathcal{C}$ with respect to \eqref{Dynamical System}. 
\begin{theorem} \label{safe if stable theorem}
Suppose there exists a continuous feedback controller $u=K(\bar{x})$ such that the origin of the safety embedded closed-loop system, $\dot{\bar{x}}=\bar{f}(\bar{x})+\bar{g}(\bar{x}) K(\bar{x})$, is asymptotically stable. Then, there exists an open neighborhood $\mathcal{D}$ of the origin such that $u=K(\bar{x})$ is safe with respect to the safety region $\mathcal{C}=\{x\in \mathcal{D}: h(x)> 0\}$. 
\end{theorem}
\begin{proof}
Let us assume that the origin of the embedded closed-loop system is asymptotically stable with a domain of attraction $\mathcal{A}$. Then, there exist open neighborhoods $\mathcal{X}\subset \mathbb{R}^n$ and $\mathcal{Z}\subset \mathbb{R}^q$ of the origin with $\mathcal{Z}$ bounded such that $\mathcal{X} \times \mathcal{Z} \subset \mathcal{A}$. Letting $\beta:=[\beta_1 \; \cdots \beta_q]^{\rm T}$, by the continuity of $\tilde{\beta}(x)=\beta(x)-\beta(0)$ on $\{x \in \mathbb{R}^n: h(x)>0\}$, the inverse image $\tilde{\beta}^{-1} (\mathcal{Z})$ of $\mathcal{Z}$ by $\tilde{\beta}$ is an open neighborhood of the origin. Thus $\mathcal{D}:= \mathcal{X} \cap \beta^{-1} (\mathcal{Z})$ is also an open neighborhood of the origin. For any initial condition $x(0) \in \mathcal{C} \subset \mathcal{D}$, it can be seen that $z(0)=\tilde{\beta}(x(0))\in \mathcal{Z}$ and $\bar{x}(0) \in \mathcal{A}$. Thus the trajectories $\bar{x}(t)$, $t\geq 0$, are bounded and converge to zero implying that $z(t)$ is bounded as well. By Propositions \ref{safety prop 1} and \ref{safety prop 2}, $\beta(x)$ is also bounded guaranteeing the safety of $u=K(\bar{x})$.
\end{proof}
The remainder of the letter is devoted to synthesizing safe feedback controllers that simultaneously stabilize \eqref{Dynamical System} and satisfy the required safety constraint \eqref{safety constraint} by stabilizing the dynamical system \eqref{new system, safety augmented} which unifies the two requirements. In the next sections, we apply the proposed methodology to constrained linear control systems to generate safe linear controls using pole placement and then to constrained optimal control problems to synthesize optimal safe controllers.

%\subsection{Linear Control Application}
\section{Safety Embedded Linear Control} \label{Safety Embedded Linear Control} 
%if this is chosen don't forget to correct the introduction, the abstract, the conclusion and everything and maybe name section III: Barrier States for Safety Embedded Control or 
%  Barrier States for Safe Stabilization
Consider the linear time-invariant system $\dot{x}=A x + B u$
subject to some safety constraint defined by a smooth function $h$. Defining the associated barrier state, augmenting it and linearizing the safety embedded model \eqref{new system, safety augmented} around the origin yields the linearized safety embedded system $\dot{\bar{x}}=\bar{A} \bar{x} + \bar{B} u$ where
\begin{gather*}
    \bar{A}=\begin{bmatrix} A & 0_{n \times 1} \\ -\gamma \phi_{1x}\big(\beta_0,h_x(0)\big) + \phi_0(\beta_0) h_x(0) A& -\gamma \phi_{1z}\big(\beta_0,h(0)\big) \end{bmatrix}  \\ 
    \bar{B}=\begin{bmatrix} B \\ \phi_0(\beta_0) h_x(0) B \end{bmatrix}
 \end{gather*}
It should be noted that this system may not be controllable. This should not pose any problem, however, as $\phi_1$ guarantees stabilizability which should be enough for us to design a safe stabilizing controller. Although we may not be able to change the location of the barrier state's pole, it is sufficient to use the states and the barrier state to construct a safe stabilizing control. 
Next, we show an example where we form a barrier state to design a safe stabilizing control while the linearized system is stabilizable but not controllable.
\subsection{Constrained Linear System Numerical Example}
Consider the open loop unstable linear system given by
\begin{equation*}
    \begin{bmatrix} \dot{x}_1 \\ \dot{x}_2 \end{bmatrix}= \begin{bmatrix} 1 & -5 \\ 0 & -1\end{bmatrix} \begin{bmatrix} x_1 \\ x_2 \end{bmatrix} + \begin{bmatrix} 0 \\ 1\end{bmatrix} u
\end{equation*}
Assume that it is desired to stay in the safe set $\mathcal{C}=\{x: (x_1 - 2)^2 + (x_2 - 2)^2 - 0.5^2 >0 \}$ and that the closed-loop system's poles are $-3$ and $-5$. Using the inverse barrier function, the linearized safety embedded system that we will use to design the linear control will be 
$$
\dot{\bar{x}}=  \begin{bmatrix} 1 & -5 & 0\\ 0 & -1 & 0 \\ \frac{4\gamma+4}{7.75^2} & \frac{4\gamma-24}{7.75^2} & -\gamma \end{bmatrix}  \bar{x} +   \begin{bmatrix} 0 \\ 1 \\ \frac{4}{7.75^2} \end{bmatrix} u
$$
where $\bar{x}=[x_1 \ \ x_2 \ \ z]^{\rm{T}}$. We use this augmented linear system to design a safe linear controller using the pole placement method to place the poles of the closed-loop controllable subsystem at $-3$ and $-5$, achieving the desired performance. When $\gamma=2$, the safe stabilizing linear control is found to be $u= -4.43 x_1 +8.38 x_2 -5.63 z$. Note that although the controller is linear with respect to the safety embedded system, it is a nonlinear function of the original state $x$ provided $z(0) = \beta(x(0)) - \beta(0)$. Fig. \ref{linear system phase portrait (top) and different gammas (bottom)} shows that indeed the designed linear controller is able to safely stabilize the system. 

\begin{figure}
        \centering
    \subfloat{\includegraphics[trim=0 0 0 0, clip, width=0.75\linewidth]{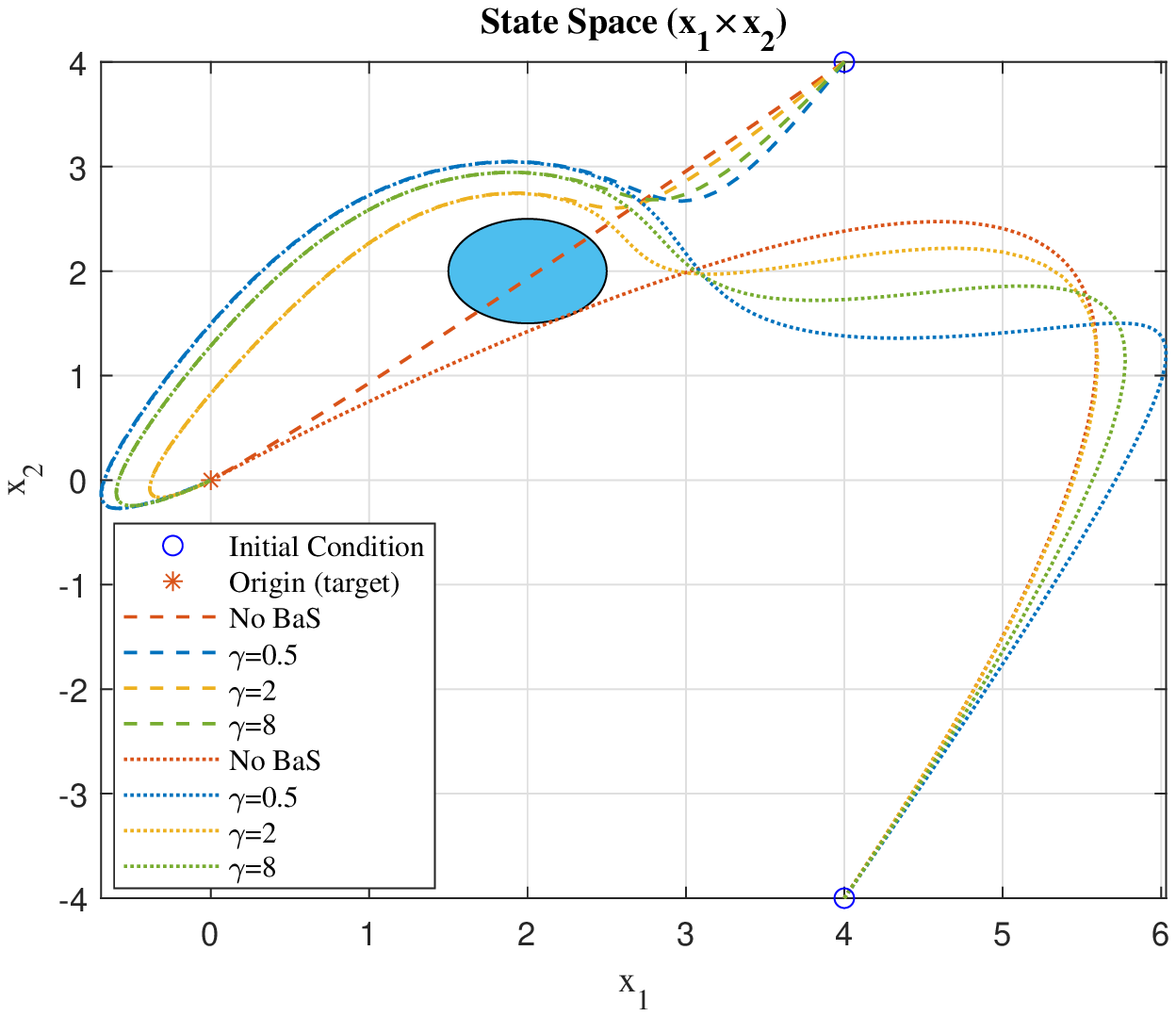}}
    \qquad
    \subfloat{\includegraphics[trim=0 3 0 0, clip, width=0.75\linewidth]{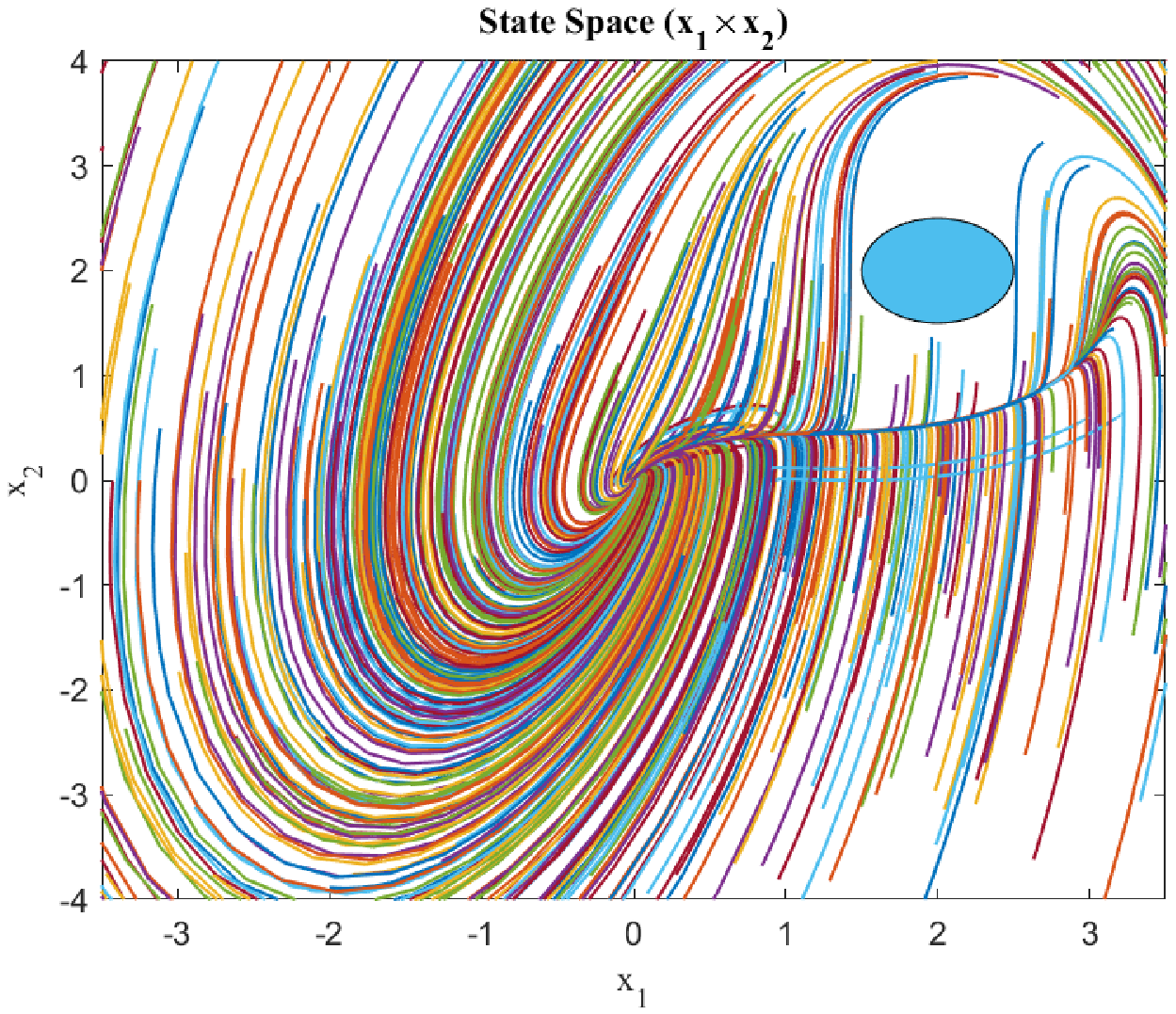}}      
   \caption{The top figure shows numerical simulations of the closed loop system starting from different initial conditions (small circles) with different $\gamma$s using the pole placement method to place the poles at $(-3,-5)$. The bottom figure shows a phase portrait of the closed loop system under the safety embedded linear control $u= -22.63 x_1 +7.14x_2 +102.96 z$ with $\gamma=1$.}
      \label{linear system phase portrait (top) and different gammas (bottom)}
\end{figure} 
   
It is worth noting that this is a linear control design for an inherently nonlinear control problem and thus limitations and difficulties of linear controls to stabilize nonlinear systems apply. One could use any nonlinear control technique to design a safely stabilizing control. In the next section, we leverage the infinite horizon optimal control in the process of designing an efficacious optimal safe control as optimal control is a well suited paradigm for such a problem and it facilitates the design of nonlinear controllers.

\section{Safety Embedded Optimal Control} \label{Sec: EMBEDDED SAFETY OPTIMAL CONTROL}
Consider the infinite horizon optimal control problem of minimizing some cost functional subject to the dynamics \eqref{Dynamical System} and the safety condition \eqref{safety constraint}. This has been a sought-after problem recently \cite{chen2020optimal,cohen2020approximate,almubarak2021HJBbased} where CBFs are used to solve constrained optimal control problems. In these efforts, it can be seen how difficult the problem is and thus various complex techniques have been developed to solve the problem. Using the proposed BaS, this problem can be solved directly using well-known unconstrained optimal control methods. A systematic approach toward solving this problem is to seek an optimal feedback controller $u=K(\bar{x})$ that minimizes
\begin{equation} \label{cost.func}
V(x(0),u(t))=\frac{1}{2} \int_{0}^{\infty} Q(\bar{x})+u^{\rm{T}} R u \; dt 
\end{equation}
where $Q:\mathbb{R}^{n+q} \rightarrow \mathbb{R}^{+} \ \forall x \neq 0$ is analytic and its Hessian is positive definite and $R \succ 0$, subject to \eqref{new system, safety augmented}. By Theorem \ref{safe if stable theorem}, if this optimal control problem is successfully solved, then both safety and stability requirements are met. For such an infinite horizon optimal control problem, a necessary and sufficient condition is that the well-known Hamilton-Jacobi-Belmman (HJB) equation is satisfied, 
\begin{equation} \label{HJB}
   \text{HJB}:= \min_{u} \ V_{\bar{x}}^*\big(\bar{f}(\bar{x})+\bar{g}(\bar{x}) u\big) + \frac{1}{2}u^{\rm{T}} R u + \frac{1}{2} Q(\bar{x}) =0
\end{equation}
with a boundary condition $V^* (0)=0$ where $V^*$ is the optimal solution, a.k.a. the value function, and $V_{\bar{x}}^* =\frac{\partial V^*}{\partial \bar{x}}$. 
\begin{theorem}
Consider the optimal control problem \eqref{new system, safety augmented}-\eqref{cost.func} with analytic $f(x), g(x)$ and $h(x)$ and suppose that
the pair $\big(\frac{\partial f}{\partial x}(0),g(0)\big)$ is stabilizable, $Q$ is analytic with positive definite Hessian and $R \succ 0$. Then, there exists a unique analytic value function $V^*(\bar{x})$ satisfying the {\rm HJB} equation \eqref{HJB}, which yields an optimal safe feedback control 
\begin{equation} \label{safe optimal control}
    u^*_{safe}(\bar{x})=- R^{-1} \bar{g}(\bar{x}) V^*_{\bar{x}} (\bar{x})
\end{equation}
Moreover, $V^*(\bar{x})$ is a Lyapunov function and $u^*_{safe}$ renders the origin of the closed loop system $\bar{f}(\bar{x})+\bar{g}(\bar{x})u^*_{safe}(\bar{x})$ asymptotically stable. Therefore, the barrier state $z$ is bounded guaranteeing the generation of safe trajectories. 
\end{theorem}
\begin{proof}
As remarked earlier in the letter, the embedded system \eqref{new system, safety augmented} preserves the continuous differentiability and stabilizability properties of \eqref{Dynamical System} and thus satisfies the analyticity and stabilizability assumptions in \cite{AlmubarakSadeghandTaylor2019,SadeghAlmubarak_recursive,lukes1969optimal} hence guaranteeing the existence and uniqueness of the value function $V^*(\bar{x})$ and the corresponding optimal controller $u^*_{safe}(\bar{x})$ described in \eqref{safe optimal control}. Furthermore, the origin of the resulting closed loop system is asymptotically stable by Lyapunov stability theory \cite{SadeghAlmubarak_recursive, khalil2002nonlinear} and by Theorem \ref{safe if stable theorem}, $u^*_{safe}(\bar{x})$ is safe which completes the proof.
\end{proof}
Various techniques have been proposed in the literature to approximate the solution to the HJB equation or the associated optimal control \cite{al1961optimal,lukes1969optimal,beard1998approximate,cimen2008state,sakamoto2008analytical,AlmubarakSadeghandTaylor2019,SadeghAlmubarak_recursive}. In the next section, we utilize these efforts to produce a power series solution of the value function and its gradient to produce the optimal safe control \eqref{safe optimal control} for the optimal control problem \eqref{cost.func}. Specifically, we mainly utilize the recursive analytic solution proposed in \cite{SadeghAlmubarak_recursive} and the nonlinear quadratic regulator (NLQR) in \cite{AlmubarakSadeghandTaylor2019}. 

\section{Numerical Implementation and Examples} \label{Sec: NUMERICAL IMPLEMENTATION}
To demonstrate the efficacy of the presented technique to produce a safe and stabilizing control, we use a single BaS to enforce safety for a model of an inverted pendulum on a cart, and multiple barrier states for a multi-mobile robot navigation task where two point-robots are asked to go to intersecting targets while avoiding collision and avoiding some unsafe region.
\subsection{Second Order Inverted Pendulum on a Cart}
This system is an unstable version with a state dependent input matrix of the mechanical system with unity parameters used in \cite{romdlony2016stabilization} to show the applicability of the Control Lyapunov–Barrier Function (CLBF) approach. The system is given by
\begin{equation*}\begin{split}
    & \dot{x}_1= x_2 \\
    & \dot{x}_2= \sin(x_1)-0.5(\tanh(10x_2)+x_2)+\cos(x_1) u
\end{split} \end{equation*}
It is desired to avoid low velocities when the angle is half way through to stabilize the pendulum at the upright position to steer clear of the $90^{\text{o}}$ angle where the system loses controllability due to the $\cos(x_1)$ term in the input matrix $g(x)$. That is, there are two symmetrical decoupled unsafe sets and therefore the safe set can be represented by two functions $h_1$ and $h_2$ such that $\mathcal{C}=\{x \in \mathcal{D} \ |\ (x_1 - 2)^2 + x_2^2 > 1 \ \cap \ (x_1 + 2)^2 + x_2^2 > 1\}$. We pick the Carroll BF with $\gamma=5$ and use equation \eqref{multi-BaS} with $z(0)=\beta(x(0))-\beta(0)$ to generate a single BaS.
To generate the optimal safe controller, it is chosen to minimize the functional \eqref{cost.func}, with $R=1, \ Q= 1 x_1^2 + 50 x_2^2 + 0.5 z$.
%$$
%R=1, \ Q=\begin{bmatrix} 0.1 & 0 \\ 0 & 50 \end{bmatrix},\ \Lambda=\begin{bmatrix} 0.1 & 0 \\ 0 & 0.1 \end{bmatrix} 
%$$ 
%\begin{equation*}
%    V(x(0),u(t))=\frac{1}{2} \int_{0}^{\infty} u^2 + x^{\rm{T}} \begin{bmatrix} 0.1 & 0 \\ 0 & 50 \end{bmatrix} x+ z^{\rm{T}} \begin{bmatrix} 0.1 & 0 \\ 0 & 0.1 \end{bmatrix} z\ dt 
%\end{equation*}
It is worth mentioning that this is an optimal control problem and thus different performances can be achieved using different cost functionals. %Moreover, the parameters $\gamma_i$ are also design parameters and play a role in determining the performance. 

Fig. \ref{pendolum simulations} shows numerical simulations for the closed-loop system under a 
$3^{\text{rd}}$ order nonlinear quadratic regulator (NLQR) \cite{AlmubarakSadeghandTaylor2019}, i.e. the approximated value function for the optimal control problem is of order four. Clearly, the proposed technique is powerful enough to generate a safe and asymptotically stable closed loop system safety is embedded in the stability of the overall system. It can be seen that there is a small cusp when the trajectories cross the $90^{\text{o}}$ angle which is a result of the lose of controllability when $\cos(90^{\text{o}})=0$. If low velocities were allowed at that specific region, that is if no barrier states are used, the closed loop system will go unstable. 

\begin{figure}
        \centering
    \includegraphics[trim=0 20 0 0, width=2.6in]{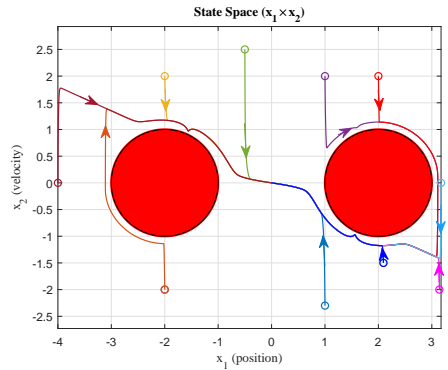}
      \caption{Numerical simulations of the closed-loop safety-critical system starting from different initial conditions (small circles). The goal is to stabilize the pendulum and avoid crossing the $90^{\text{o}}$ angle with low velocities, where the unsafe region is represented by the red circles. The proposed technique successfully generates safe trajectories and safely stabilizes the origin.}
      \label{pendolum simulations}
   \end{figure}

\subsection{Multi Simple Mobile Robot Collision Avoidance} 
In this example, two simple mobile robots are asked to navigate their way toward prespecified targets. The robots are to avoid colliding as well as avoid an obstacle on their way. The robots dynamics are
\begin{equation*} \begin{split}
    \dot{x}_i= \begin{bmatrix} u_{i1} \\ u_{i2}\end{bmatrix}, \dot{x}_j= \begin{bmatrix} u_{j1} \\ u_{j2}\end{bmatrix}
\end{split}  \end{equation*}
To avoid collision, a BaS is featured through a barrier function that prevents the robots from getting too close to each other. We pick the maximum distance between the two robots to be $\delta=0.1$. Therefore, the associated safe set is $\{x_i,x_j \in \mathcal{D}\ | \ ||x_{i}-x_{j}||^2 > \delta^2 \}$. Furthermore, we add an obstacle which is represented by a circle at $(0,0)$ with a radius of $0.25$. This calls for a BaS for each agent. Hence, the overall safe set for each agent is given by $\mathcal{C}=\{x_k\in \mathcal{D}\ | \ ||x_{i}-x_{j}||^2 > \delta \; \text{and} \; x_{k1}^2 + x_{k2}^2>0.25^2, \;  k=i,j \}$. %Therefore, the functions that define the safe set are given by
%\begin{equation*}
%    \begin{split}
%        & h_1(x_i,x_j)=||x_{i}-x_{j}||-\delta^2 \\
%        & h_k= (x_{k1}+x_{k2})^2 - 0.25^2 \\
        %& h_3= (x_{j1}+x_{j2})^2 - 0.25^2 
%    \end{split}
%\end{equation*}
%for $k=i,j$. 
For this example, we select the Log BF, $\beta_l=\log(\frac{1+h_l}{h_l})$ for $l=1,2,3$, to construct three barrier states, where the first represents the distance constraint between the two robots and the second and third barrier states represent the barriers needed to avoid the obstacle for agent $i$ and agent $j$ respectively. %It is worth noting that there is no restriction on choosing a mix of different barrier functions. 
Using Table \ref{tab:phi}, the barrier states are given by
\begin{equation*}
     \dot{z}_l= -\gamma_l \big( h_l (e^{z_{cl}}-1)^2 - e^{z_{cl}} +1 \big) - 4 \sinh^2(z_{cl}/2)h_{lx} u_l
\end{equation*}
%\begin{equation*}
%    \dot{z}=\begin{bmatrix} -\gamma_1 \big( h_1 (e^{z_{c1}}-1)^2 - e^{z_{c1}} +1 \big) - 4 \sinh^2(z_{c1}/2)h_{1x} \boldsymbol{u} \\ 
%    -\gamma_2 \big( h_2 (e^{z_{c2}}-1)^2 - e^{z_{c2}} +1 \big) - 4 \sinh^2(z_{c2}/2)h_{2x} [u_{i1},u_{i2}]^{\rm{T}} \\
%    -\gamma_3 \big( h_3 (e^{z_{c3}}-1)^2 - e^{z_{c3}} +1 \big) - 4 \sinh^2(z_{c3}/2)h_{3x} [u_{j1},u_{j2}]^{\rm{T}}
%            \end{bmatrix}
%\end{equation*}
with $z_l(0)=\beta_l(x(0))-\beta_l(0)$ for $l=1,2,3$ where $\gamma_1=15,\gamma_2=\gamma_3=0.5, z_{cl}=(z_l+c_l), c_l=\log(\frac{1+h_l(0)}{h_l(0)})$, $u_1=[u_i^{\rm{T}},u_j^{\rm{T}}]^{\rm{T}}$, $u_2=u_i$ and $u_3=u_j$. The cost functional \eqref{cost.func} is selected such that $R=I_{4}$ and $Q=x_i^{\rm{T}}x_i + x_j^{\rm{T}}x_j + 0.001z_1^2 + 0.5 z_2^2 + 0.5z_3^2$.
%$$
%R=I_{4}, \ Q=I_{4},\ \Lambda=\begin{bmatrix} 0.001 & 0 & 0 \\ 0 & 0.5 & 0 \\ 0 & 0 & 0.5 \end{bmatrix}
%$$ 
%where $I_{4}$ is the identity matrix of size $4$. 
%\begin{equation*}
%    V(x(0),u(t))=\frac{1}{2} \int_{0}^{\infty} u^{\rm{T}}I_{4}u + 5 x^{\rm{T}} I_{4} x + z^{\rm{T}} \begin{bmatrix} 0.1 & 0 & 0 \\ 0 & 0.5 & 0 \\ 0 & 0 & 0.5 \end{bmatrix} z \ dt 
%\end{equation*}
A $3^{\text{rd}}$ order NLQR is used in this example. As shown in Fig. \ref{Robots with obstacle}, using the proposed technique, we are effectively able to send the two robots safely to the targeted positions while avoiding the obstacle as well as colliding with each other. The robots get very close to each other at sometime but never get too close, i.e. the distance between them is never less than or equal to $\delta$. %With the given choice of design parameters, we are allowing the two robots to get really close to each other but not the obstacle. This was made so that the robots take symmetrical trajectories and react to the obstacle first then try to avoid getting too close and potentially collide. 
%It is worth noting that one can adjust $\gamma$s or the $\Lambda$ to force the robots to stay further away. %For example, when $\gamma_1$ is very small, the two robots steer very far away from each other and in some choices never cross same points. This also happens when we heavily penalize $z_1$ and lightly penalize $z_2$ and $z_3$. This would also mean a sharper curve in  Fig. \ref{Distance between Robots}.

 \begin{figure}
        \centering
     \subfloat{\includegraphics[trim=0 0 0 0, clip, width=0.75\linewidth]{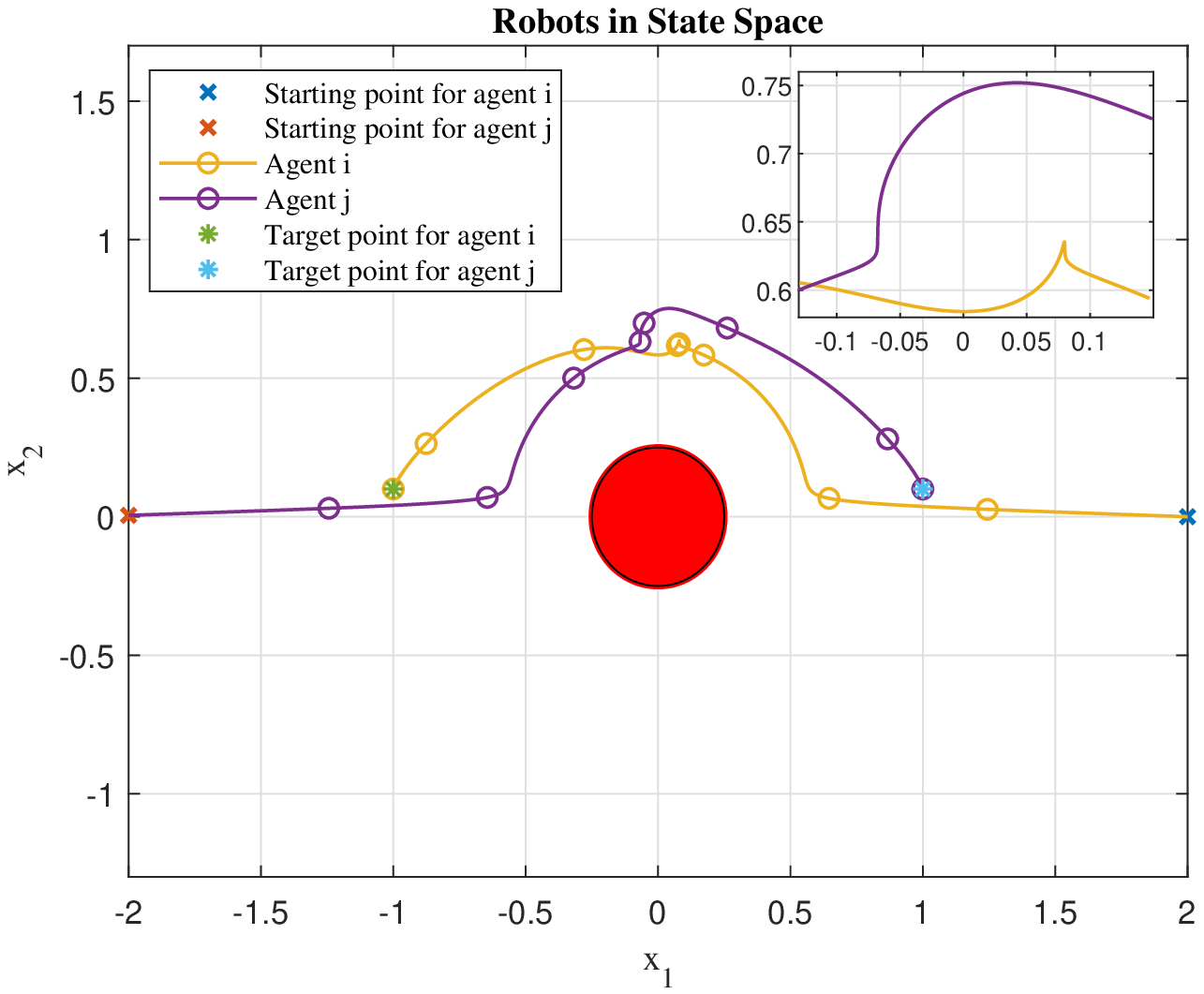}}
     \qquad
    \subfloat{\includegraphics[trim=0 5 0 0, clip, width=0.75\linewidth]{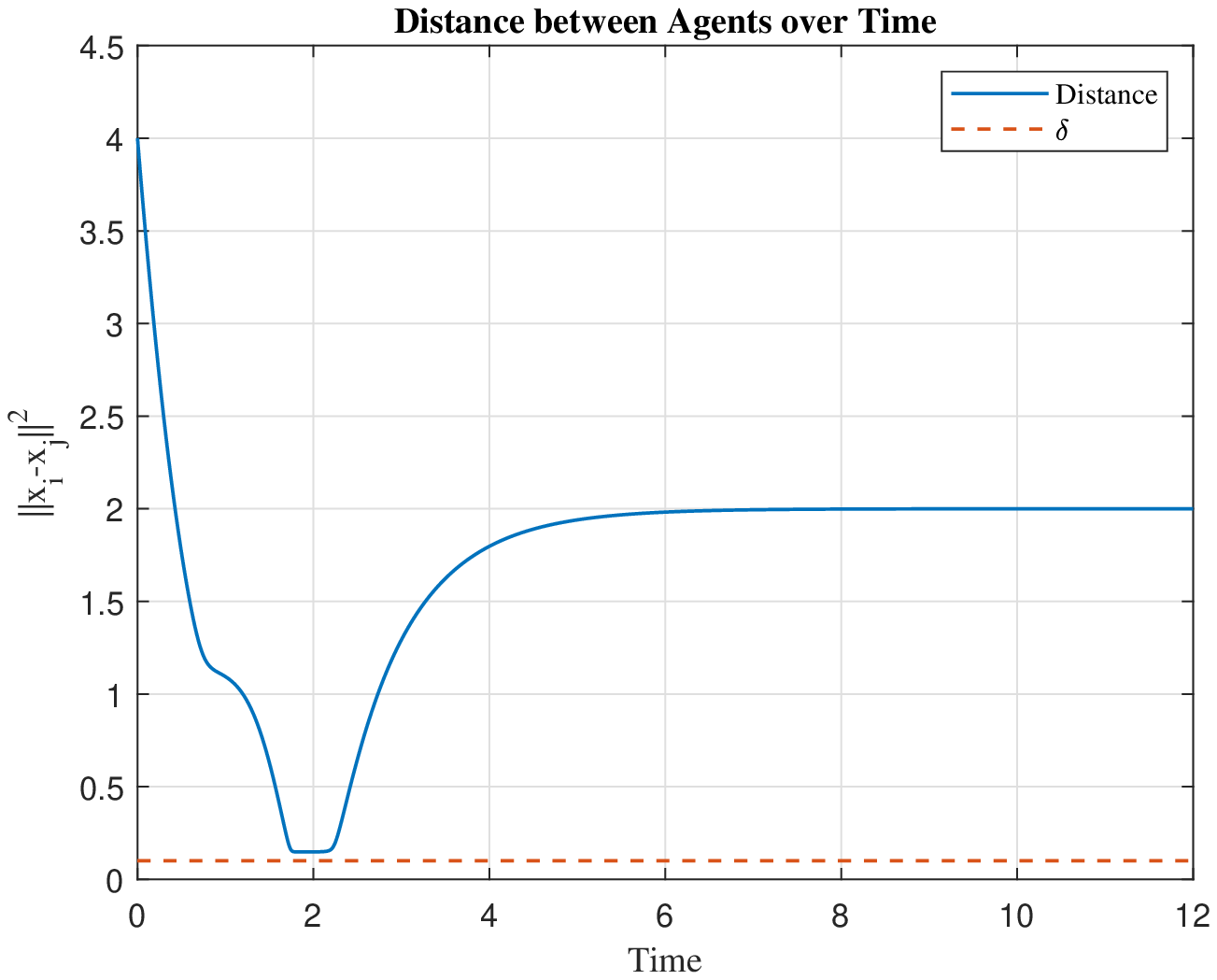}}      \caption{Two mobile robots starting from two opposite points sent to two opposite target points with an obstacle in-between. The two robots arrive at the desired locations safely. It can be seen that the two robots start by moving in a straight line toward their target points, which is the optimal behavior, but start to steer away to avoid the obstacle and then take some turns to avoid crossing the same point at the same time and to avoid getting too close (less than $\delta$ distance) to each other as shown in the bottom figure.}
      \label{Robots with obstacle}
   \end{figure} 

%The successful simulations show the validity of the proposed method and its efficacy in handling safety constraints.
\section{Conclusions and Future Works} \label{Sec: CONCLUSIONS}
A novel construction of safety embedded controls through the development of barrier states was presented. Through a proper conversion of barrier functions, barrier states were augmented to generate a nominal model which is safe if is asymptotically stable. Using barrier states, the constrained control problem was transformed to an unconstrained control problem, which makes the safety problem easier to be considered with various control techniques such as optimal control. Moreover, the BaS method is agnostic to the relative degree of the function describing the safe set with respect to the system unlike existing CBF based methods. Furthermore, there is no need to relax the stability requirement to guarantee safety as the two conditions are coupled and achieved simultaneously. The disadvantages of this approach include increasing the dimension of the model and adding nonlinearity to the model. The safety embedded model was used in constrained linear control and in the context of optimal control to generate safe stabilizing controllers. Simple linear controls and nonlinear quadratic controls were used to show the efficacy of the proposed method in various simulation examples. 
%The first example exemplified a pendulum on a cart model where a loss of controllability occurs in some regions. The effectiveness of the proposed technique was shown by successfully avoiding low velocity regions that were defined as unsafe. The second example was a multi-robot example where two simple mobile robots are asked to safely go to opposite potentially unsafe targets, i.e. while avoiding obstacles on the way, if any, and without colliding with each other.
 
Future work will include generalizations to nonlinear stochastic systems and applications to infinite and finite horizon stochastic optimal control as well as sampling-based model predictive control formulations. Another line of research includes generalizations of the proposed optimal control framework to min-max and  $H_{\infty}$ optimal control problem formulations. Finally, incorporating uncertainty quantification methods into the augmented state space representation in \eqref{new system, safety augmented} is an active research.  

%%%%%%%%%%%%%%%%%%%%%%%%%%%%%%%%%%%%%%%%%%%%%%%%%%%%%%%%%%%%%%%%%%%%%%%%%%%%%%%%
%%%%%%%%%%%%%%%%%%%%%%%%%%%%%%%%%%%%%%%%%%%%%%%%%%%%%%%%%%%%%%%%%%%%%%%%%%%%%%%%

%%%%%%%%%%%%%%%%%%%%%%%%%%%%%%%%%%%%%%%%%%%%%%%%%%%%%%%%%%%%%%%%%%%%%%%%%%%%%%%%

%\addtolength{\textheight}{-12cm}   % This command serves to balance the column lengths
                                  % on the last page of the document manually. It shortens
                                  % the textheight of the last page by a suitable amount.
                                  % This command does not take effect until the next page
                                  % so it should come on the page before the last. Make
                                  % sure that you do not shorten the textheight too much.

%%%%%%%%%%%%%%%%%%%%%%%%%%%%%%%%%%%%%%%%%%%%%%%%%%%%%%%%%%%%%%%%%%%%%%%%%%%%%%%%

%%%%%%%%%%%%%%%%%%%%%%%%%%%%%%%%%%%%%%%%%%%%%%%%%%%%%%%%%%%%%%%%%%%%%%%%%%%%%%%%

%%%%%%%%%%%%%%%%%%%%%%%%%%%%%%%%%%%%%%%%%%%%%%%%%%%%%%%%%%%%%%%%%%%%%%%%%%%%%%%%

%\bibliographystyle{plain}  % Include this if you use bibtex 
%\bibliography{Embed_Safe_cit}           % and a bib file to produce the 
%\bibliographystyle{authordate1} % bibliography (preferred). The
%\bibliographystyle{ieee}

%\begin{thebibliography}{99}
%\end{thebibliography}

\printbibliography
\end{document}